\documentclass[journal,12pt,draftclsnofoot,onecolumn]{IEEEtran}
\usepackage{mathrsfs}
\usepackage{amsmath}
\usepackage{amsthm}
\usepackage{tabularx}
\usepackage{booktabs}
\usepackage{amssymb}
\usepackage{graphicx}
\usepackage{subfigure}
\usepackage{times}
\usepackage{latexsym}
\usepackage{subfigure}
\usepackage{cite}
\usepackage{balance}
\usepackage{multirow}
\usepackage{algorithmic}
\usepackage{color}
\usepackage{float}
\usepackage{cases}
\usepackage{algorithm}
\usepackage{yhmath}
\usepackage{bm}
\usepackage{bbm}
\usepackage{dsfont}
\usepackage{footnote}
\usepackage{makecell}
\usepackage{url}
\usepackage[justification=centering]{caption}
\newtheorem{theorem}{Theorem}

\theoremstyle{remark}
\newtheorem{assumption}{Assumption}
\newtheorem{remark}{Remark}

\begin{document}
\title{\Large CB-DSL: Communication-efficient and Byzantine-robust \\Distributed Swarm Learning on Non-i.i.d. Data}
\author{\normalsize Xin Fan$^{1}$,~\IEEEmembership{\normalsize Student Member,~IEEE}, Yue Wang$^2$,~\IEEEmembership{\normalsize Senior Member,~IEEE}, \\ Yan Huo$^{1}$,~\IEEEmembership{\normalsize Senior Member,~IEEE}, and Zhi Tian$^2,~\IEEEmembership{\normalsize Fellow,~IEEE}$\\
$^{1}$School of Electronics and Information Engineering, Beijing Jiaotong University, China\\
$^2$Department of Electrical \& Computer Engineering, George Mason University, USA
\\E-mails: \{yhuo,fanxin\}@bjtu.edu.cn, \{ywang56,ztian1\}@gmu.edu}

\maketitle

\vspace{-0.5in}
\begin{abstract}\vspace{-0.1in}
The valuable data directly collected by Internet of Things (IoT) devices in edge networks together with the resurgence of machine learning (ML) stimulate the latest trend of artificial intelligence (AI) at the edge. However, traditional ML and recent federated learning (FL) methods face major challenges including communication bottleneck, data heterogeneity and security concerns in edge IoT scenarios, especially when being adopted for distributed learning among massive IoT devices equipped with limited data and transmission resources. Meanwhile, the swarm nature of IoT systems is overlooked by most existing literature, which calls for new designs of distributed learning algorithms. In this paper, inspired by the success of biological intelligence (BI) of gregarious organisms, we propose a novel edge learning approach for swarm IoT, called  communication-efficient and Byzantine-robust distributed swarm learning (CB-DSL), through a holistic integration of AI-enabled stochastic gradient descent and BI-enabled particle swarm optimization. To deal with non-independent and identically distributed (non-i.i.d.) data issues and Byzantine attacks, {\color{black}a very small amount of} global data samples are introduced in CB-DSL and shared among IoT workers, which not only alleviates the local data heterogeneity effectively but also enables to fully utilize the exploration-exploitation mechanism of swarm intelligence. Further, we provide convergence analysis to theoretically demonstrate that the proposed CB-DSL is superior to the standard FL with better convergence behavior. In addition, to measure the effectiveness of the introduction of the globally shared dataset, we also 
{\color{black} evaluate the model divergence by deriving its upper bound, which is related to} 
the distance between the data distribution at local IoT devices and the population distribution for the whole datasets. Numerical results verify that the proposed CB-DSL outperforms the existing benchmarks in terms of faster convergence speed, higher convergent accuracy, lower communication cost, and better robustness against non-i.i.d. data and Byzantine attacks\footnote{\textcolor[rgb]{0.00,0.00,0.00}{Our code can be found at:\url{https://github.com/fuanxiyin/CB-DSL.git}.}}.
\end{abstract}
\vspace{-0.1in}
\begin{IEEEkeywords}\vspace{-0.1in}
Distributed swarm learning, federated learning, particle swarm optimization, non-i.i.d. data, convergence analysis, model divergence analysis.
\end{IEEEkeywords}

\section{Introduction}
With the vigorous growth of Internet of Things (IoT) and Internet of Vehicles (IoV), smart devices are becoming the workhorse at the edge of wireless networks beyond 5G (B5G) \cite{Wang2018Big}.
The valuable data directly collected by devices together with the resurgence of machine learning (ML) stimulate the latest trend of artificial intelligence (AI) at the edge of B5G networks, termed as edge learning or edge AI \cite{wang2019edge,shi2020communication}. When conventional ML techniques are applied for edge learning, they are typically deployed in a centralized mode, which hinges on a full collection of distributed local data from the edge IoT devices to a central node. Such a centralized learning approach can obtain high learning accuracy, but the raw data collection process not only consumes huge communication resources but also raises unwilling privacy exposure and severe security concerns~\cite{IoTprivacy}. Alternatively, federated learning (FL) has recently attracted great attention and resulted in fruitful attempts for learning-based applications among multiple distributed workers such as personal mobile phones, which allows distributed learning from local data without raw data exchange\cite{konevcny2016federated,mcmahan2017communication}.

Standard FL methods are originally designed for ideal learning settings and wireless environments, which however face several challenges when being adopted for distributed learning among massive edge IoT devices that are usually equipped with limited capability and resources.
As the number of model parameters goes very large in deep neural networks, transmission of all the local model updates in FL between IoT devices (working as local workers) and the parameter server (PS) incurs high communication overhead in edge networks.
Further, stochastic gradient descent (SGD) is widely applied for model training in FL \cite{mcmahan2017communication,qian2015efficient}, where independent and identically distributed (i.i.d.) data samples are assumed at local workers and transmission is assumed error-free in order to ensure unbiased estimates and good empirical performances \cite{gower2019sgd,bottou2010large}.
However, in edge IoT scenarios, local training data samples at different IoT workers turn to be statistically heterogeneous worker-by-worker, giving rise to the non-i.i.d. data issue that may considerably degrades the learning performance of standard FL methods, e.g., Federated Averaging (FedAvg)~\cite{li2020federated,zhao2018federated}.
In addition, gradient{\color{black}-based} algorithms are subject to local optimum traps in solving non-convex problems \cite{boyd2004convex,huo2016asynchronous,vlaski2021second}, such as when training neural networks with nonlinear activations. {\color{black}This issue is aggravated in distributed settings,
especially when local workers only collect small-volume data.} 
Last but not the least, standard FL performs well in attack-free network settings, but is vulnerable to Byzantine attacks that may exist in practical edge networks~\cite{yang2020adversary,fan2022bev,yin2018byzantine,Xin2022iccws}.

Although some of the aforementioned challenges have been recently investigated in the literature of FL for edge networks and IoT applications \cite{khanna2020internet,lim2020federated,hussain2020machine}, they mainly focus on the modification and customization of the standard FL techniques, which however largely neglect some important and unique characteristics of IoT devices in edge networks.
Such unique characteristics include the large population of devices for many IoT applications, limited communication bandwidth available in edge networks, and non-i.i.d. local data with small data volume at individual IoT workers.
By ignoring these characteristics, existing efforts on edge learning fail to consider these limitations in the learning algorithm design for edge IoT systems, {\color{black}which results in learning performance degradation of FL applied to practical IoT edge networks}.
On the other hand, biological organisms in nature have demonstrated swarm intelligence with superior strength in collectively processing information, making decisions, dealing with uncertainties, adapting to environment changes, and recovering from errors and failures, even though they are individually weak.
All these attributes of biological intelligence (BI) are desired by IoT edge learning systems.
Notably, bio-inspired swarm optimization techniques are good at collaboratively finding the globally optimal solutions to complex optimization problems thanks to their built-in exploration-exploitation mechanism in swarms, but their convergence speed is typically slow \cite{selvaraj2014survey,almufti2019taxonomy}.

Motivated to bridge these gaps, this paper leverages both AI and BI to develop a communication-efficient and Byzantine-robust distributed swarm learning (CB-DSL) approach, by reformulating the bio-inspired particle swarm optimization (PSO) problem as a distributed learning problem with non-i.i.d. local data {\color{black}and in the presence of malicious attacks}.
For non-convex problems, by taking advantage of the exploration-exploitation mechanism of PSO\cite{eberhart1995new,kennedy1995particle}, our CB-DSL solutions have an increased chance to jump out of local optimum traps via swarm intelligence.  
For the communication bottleneck challenge, our CB-DSL only requires the best worker having the 
minimum loss function value to upload its local model to the PS, which thus dramatically reduces the communication overhead and energy consumption in edge networks.
To alleviate the non-i.i.d. data issue, we propose to introduce a small{\color{black}-volume} global dataset that is shared among all local workers for dual purposes. 
A part of this {\color{black}globally shared dataset} 
is used 
for training, whose effectiveness in relieving the non-i.i.d. problem is evaluated through the model divergence analysis.
The other part of the global dataset is used {\color{black}to calculate the fair-value loss} for scoring the local models. 
{\color{black}It} helps to identify the per-worker best model for best worker selection{\color{black}, and enables} to
verify the uploaded local model by which the PS can screen Byzantine attackers. 
Our main contributions are summarized as follows.
\begin{itemize}
\item We propose a new CB-DSL framework {\color{black}by developing a
holistic integration of AI-driven SGD and BI-driven PSO}, to {\color{black}effectively} 
handle the high communication costs, non-i.i.d. issues, non-convex problems and Byzantine attacks without sacrifice convergence speed, which cannot be achieved by SGD or PSO alone.
    {\color{black}CB-DSL offers a new paradigm of efficient and robust edge learning tailored}  
    for massive smart IoT devices in edge networks, which brings the 
    {\color{black}benefits of swarm intelligence} to broad applications of distributed learning.
 \item From {\color{black}the} theoretical point of view, we {\color{black}are the first one to systematically analyze the combination of FL and PSO, by deriving a}  closed-form expression {\color{black}to quantify the 
 expected convergence rate achieved} by our CB-DSL. Our analytical results 
     not only reflect the impact of different settings and parameters of our CB-DSL on the performance of edge learning among distributed workers, but also indicate that our CB-DSL outperforms the standard FL methods such as FedAvg {\color{black}in terms of better convergence rate}. 
 \item We further investigate the non-i.i.d. data issue at distributed workers by providing a model divergence analysis to evaluate how the introduction of a globally shared dataset improves the learning performance of our CB-DSL. 
     {\color{black}Our theoretical result reveals that the model divergence is subject to an upper bound, 
     which is decided by} 
     the earth mover's distance (EMD) between the data distribution at local workers and the population distribution for the whole datasets. 
 \item Through comprehensive experiments, we test the proposed CB-DSL approach in solving image classification problems by using the MNIST dataset. Simulation results show that our CB-DSL outperforms the benchmark methods in terms of achieving {\color{black}the highest} 
     testing accuracy with {\color{black}the fastest} 
     convergence under both the i.i.d. and non-i.i.d. cases and even in the presence of Byzantine attacks.
\end{itemize}

The rest of this paper is organized as follows. Section II reviews the related work. The problem formulation of distributed learning and the framework of CB-DSL technique are systematically presented in Section III, where we develop the CB-DSL algorithm.
The 
expected convergence rate and the model divergence analysis of the CB-DSL technique are studied  
in Section IV and Section V, respectively. 
Section VI presents simulation results and comparison of the CB-DSL technique with the benchmark methods, 
followed by conclusions in Section VII.

\emph{Notations:} Bold upper and lower case letters denote matrices and vectors, respectively. Euclidean norm of a vector or a matrix is depicted as $\|\cdot\|$.  The expectation and the first order derivative are represented by $\mathbb{E}$ and $\nabla$. {\color{black}$\langle\cdot , \cdot\rangle$ calculates the inner product of two vectors.} The probability of an event $y=c$ is expressed as $p(y=c)$. The event indicator is symbolled as  $\mathds{1}_{y=c}$, which is equal to $1$ when $y=c$, or $0$ otherwise.

\section{Related Work}\label{Sec:Related}

Various methods have been proposed in addressing the communication challenges of FL, such as sparsification\cite{aji2017sparse,lin2018deep},
quantization\cite{liu2019decentralized,seide20141,alistarh2017qsgd} and infrequent uploading of local updates \cite{liu2019communication,xu2019coke,xu2018energy,xu2021coke}. Theses methods aim to reduce the amount of the communication overhead, by either compressing or dropping some non-informative transmissions. 
These strategies are investigated predominantly for FL over digital channels based on the orthogonal transmission resource allocation among different local workers. 
Recently, a promising technique for tackling the communication bandwidth bottleneck emerges in the form of FL over the air \cite{yang2020federated,Xin2022icc,fan2021joint},
which exploits the fact that the model aggregation operation in FL matches the waveform superposition property of the wireless analog multi-access channels. 
{\color{black}It can further incorporate other efficiency-enhancing strategies for effectively reducing bandwidth consumption.} 
For instance,
communication-efficient FL over the air is developed by combining compression, quantization and concurrent transmission through 1-bit compressive sensing and analog aggregation transmission in \cite{fan2021communication,fan20211}. 
Nevertheless, the aforementioned compression and transmission strategies still require all participating workers to
exchange some variants of
their local updates,
which are not tailored for edge IoT systems and may result in tremendous communication costs and energy consumption in edge networks with massive IoT devices.

To take advantage of the swarm biological intelligence of animal flocks, particle swarm optimization (PSO) has been developed to solve complicated optimization problems without invoking {\color{black}the assumption on convexity~\cite{eberhart1995new,kennedy1995particle}.}
Recently,
there are few attempts of applying the PSO ideas to improve
machine learning performance.
In the centralized setting, PSO is used to optimize the solution and hyperparameters of convolutional neural networks (CNNs) for enhanced recognition accuracy of image classification\cite{syulistyo2016particle,junior2019particle, serizawa2020optimization}. 
In the distributed setting, 
two relevant works are found in attempting to integrate PSO into FL to improve FL performance \cite{qolomany2020particle, park2021fedpso}.
In \cite{qolomany2020particle}, 
FL is used for learning, while PSO is simply applied to search the optimal hyperparameters.
Different from \cite{qolomany2020particle}, our work focuses on the design of distributed training algorithm and model updating strategy for improving the performance and robustness {\color{black}of edge learning systems}. 
In \cite{park2021fedpso}, PSO and FL are combined in a simplistic manner for the idealized distributed settings with i.i.d. data and no attacks, which cannot be guaranteed for practical edge IoT systems.
Further, the work \cite{park2021fedpso} actually builds on an implicit assumption that a common loss function is available to all distributed workers, which trivializes the assessment of the globally best model. However, in distributed learning problems, loss function is only partially observable at local workers, which is data-dependent and hence different across workers. Thus, the method in \cite{park2021fedpso} is not suitable to edge IoT systems with 
data of small volume at local workers.
More importantly, there has not been any work on theoretical analysis for 
performance evaluation and convergence guarantee {\color{black}of distributed learning by connecting PSO with FL}.
{\color{black}To fill these identified technical gaps, in the next sections, we develop a novel efficient and robust edge learning algorithm through a holistic integration of AI-driven SGD and BI-driven PSO and empowered by using a small-volume global dataset,
whose advantages are verified by rigorous convergence analysis, model divergence evaluation, and experiments on real data.}  

\section{Distributed Swarm Learning}\label{sec:Model}
This section starts with the problem statement for distributed learning and the formulations of FL and PSO techniques. 
Then, the pros and cons of FL and PSO motivate us to bridge distributed learning with swarm optimization techniques to make the best use of both artificial and biological intelligence.
In particular, we focus on {\color{black}a systematical integration of FL and PSO for} 
a novel communication-efficient and Byzantine-robust edge learning algorithm with non-i.i.d. local data {\color{black}and in the presence of Byzantine attacks}.

Consider a distributed 
learning model with one parameter server (PS) and $U$ IoT workers,
where $U$ is very large but each worker has data of small volume 
in edge IoT scenarios.
Assume that each worker has $K_i$ data samples in its local dataset $\mathfrak{D}_i$, with
$|\mathfrak{D}_i|=K_i$, and
$i=1,\dots,U$. Denote $(\mathbf{x}_{i,k},y_{i,k})$ as the $k$-th data sample of the $i$-th local worker. Let $f(\mathbf{w};\mathbf{x}_{i,k},y_{i,k})$ represent the loss function associated with each data sample $(\mathbf{x}_{i,k},y_{i,k})$, where $\mathbf{w}=[w^1, \ldots, w^D]$ of size $D$ consists of the parameters of a common learning model.
The corresponding population loss function for the whole datasets $\mathfrak{D}$ and that for the local dataset $\mathfrak{D}_i$ of the $i$-th worker are denoted as $F(\mathbf{w}):=\mathbb{E}_{\mathfrak{D}}[f(\mathbf{w};\mathbf{x}_{i,k},y_{i,k})]$ and $F_i(\mathbf{w}):=\mathbb{E}_{\mathfrak{D}_i}[f(\mathbf{w};\mathbf{x}_{i,k},y_{i,k})]$, respectively, where $\mathfrak{D}=\bigcup_i \mathfrak{D}_i$.
For distributed learning, local workers collaboratively learn $\mathbf{w}$ by minimizing
\begin{equation}\label{eq:lossfopt}
 \textbf{P1:} \quad \mathbf{w}^*_i= \arg \min_{\mathbf{w}_i}  \, F_i(\mathbf{w}_i), \quad \;\,\
 \mbox{s.t.},  \quad \mathbf{w}_i = \mathbf{z},\quad \forall i,
\end{equation}
where $\mathbf{z}$ is an auxiliary variable to enforce consensus through collaboration among distributed local workers.


\subsection{Federated Learning}
For standard FL designed in ideal learning settings and network environments \cite{konevcny2016federated},
the minimization of $F_i(\mathbf{w})$ is typically carried out by the stochastic gradient descent (SGD) algorithm\cite{konevcny2016federated,mcmahan2017communication}, 
where local workers iteratively update their local models in FL as
\begin{align}\label{eq:localupdate0}
    \mathbf{w}_{i,t+1} = \mathbf{w}_{i,t}-\textstyle\frac{\alpha}{U} \textstyle\sum_{j=1}^{U}\nabla F_j(\mathbf{w}_{t};\mathbf{x}_{j,k},y_{j,k}),
\end{align}
where $\alpha$ is the learning rate and $\nabla F_j(\mathbf{w}_{t};\mathbf{x}_{j,k},y_{j,k})= \mathbb{E}_{\mathfrak{D}_j}\left[\frac{\sum_{\mathfrak{B}_j}\nabla f(\mathbf{w}_{t};\mathbf{x}_{j,k},y_{j,k})}{|\mathfrak{B}_j|}\right]$ is the local gradient computed at each local worker using its randomly selected mini-batch $\mathfrak{B}_j \subset \mathfrak{D}_j$ with the mini-batch size $|\mathfrak{B}_j|$.

Note that \eqref{eq:localupdate0} is the mathematical illustration of the iterative local model update, whereas the second term of global gradient averaging therein is typically implemented at the PS and then sent back to local workers. Hence, communications take place in every iteration until convergence, during which the communication overhead to acquire the sum of all $U$ local gradients in \eqref{eq:localupdate0} would be huge especially when $U$ and $D$ are large.
Moreover, for complicated non-convex problems, distributed gradient-based FL solutions may converge to undesired local optima and there is unfortunately a lack of effective mechanisms to escape these traps.

\subsection{Particle Swarm Optimization}
As a bio-inspired algorithm, PSO is a stochastic optimization approach based on the movement of particles (workers) and the collaboration of swarms to iteratively and cooperatively search for an optimal solution to general optimization problems \cite{eberhart1995new,kennedy1995particle}. 
{\color{black}Note that PSO is originally designed for optimization problems instead of learning problems with distributed data. In this sense, the loss function in PSO is assumed to be globally common to all particles, i.e., $F_i(\cdot) = F(\cdot),  \forall i$ in the problem \textbf{P1} in \eqref{eq:lossfopt}. This is however not the case in distributed leaning where $F_i(\cdot)$ is data-dependent and different worker-by-worker, which will be explained in the next subsection.}

In PSO, a swarm consists of a large set of particles, $i=1,2...,U$. At the current iteration, the position $\mathbf{w}_{i,t}$ of each particle $i$ presents a possible solution to the problem, and meanwhile the velocity $\mathbf{v}_{i,t}$ of each particle $i$ denotes the updating direction for the next step.
To find the globally optimal value of $F(\cdot)$, particles collaborate with each other to update their velocities and positions in an iterative manner
\begin{align}
\mathbf{v}_{i,t+1}&=c_0 \mathbf{v}_{i,t}+c_1(\mathbf{w}_{i,t}^p-\mathbf{w}_{i,t})+c_2 (\mathbf{w}_t^g-\mathbf{w}_{i,t}), \label{eq:PSO_v}\\
\mathbf{w}_{i,t+1}&=\mathbf{w}_{i,t}+\mathbf{v}_{i,t+1}, \label{eq:PSO_w}
\end{align}
where the velocity is updated as a combination of three sub-directions: inertia $\mathbf{v}_{i,t}$ of the previous updating direction, individual direction towards each particle's own historical best parameter $\mathbf{w}^p_{i,t} {=} \arg\!\min_{\tau{=}1,{\cdots},t} F(\mathbf{w}_{i,\tau})$, and social direction towards the globally best parameter found by the entire swarm $\mathbf{w}^g_{t} {=} \arg\! \min_{i{=}1,{\cdots},U} F(\mathbf{w}^p_{i,t})$.
Among the corresponding three weights, the inertia weight $c_0$ is a positive number, while $c_1$ and $c_2$ are positive and random (say, uniformly distributed as $c_1{\sim} \mathcal{U}(0, \delta_{c_1})$, and $c_2{\sim} \mathcal{U}(0, \delta_{c_2})$) for stochastic optimization.
Notably, the weighted combination of the three sub-directions in \eqref{eq:PSO_v} serves a mechanism for exploration-exploitation tradeoffs, where $c_0$ is set to be linearly decreasing over iterations to tune the solution search process from exploration to exploitation, and $c_1$ and $c_2$ indicate the random exploration level at individual particles and the exploitation level in swarm, respectively.

\subsection{Communication-efficient and Byzantine-robust Distributed Swarm Learning} 
A major {\color{black}challenge} 
from optimization problems to learning problems with distributed data is the lack of a common $F(\cdot)$ for global assessment, which however becomes $F_i(\cdot; \mathfrak{D}_i)$ dependent on local dataset $\mathfrak{D}_i$ in distributed learning.
{\color{black}Facing this challenge, we first introduce a very small amount of
global dataset\footnote{For the implementation point of view, a small amount {\color{black}(e.g., 1\% of all datasets is adequate as used} in our simulations) of globally shared dataset can be generated by a generative adversarial network module for keeping the privacy of workers' own local data\cite{wang2017generative}, which can be either pre-stored in the IoT devices or broadcasted from the PS to all the local workers. {\color{black}The required resources in sharing and local storage are quite low.}}:} $\mathfrak{D}^G=\mathfrak{D}^G_{tr}\cup\mathfrak{D}^G_{sc}$  to be shared by all workers{\color{black}, and then propose} a novel edge learning framework called communication-efficient and Byzantine-robust distributed swarm learning (CB-DSL).
The CB-DSL algorithm is implemented in \textbf{Algorithm~\ref{alg:policyforCBFedPSO}}, and schematically illustrated through the following iterative model updating steps.

At the local workers $i = 1, \cdots, U$, the model parameters are updated in a way of integrating BI-enabled PSO with AI-enabled SGD
\begin{equation}\label{eq:DSL_w}
 \mathbf{w}_{i,t+1}=\mathbf{w}_{i,t}+ \underbrace{c_0 \mathbf{v}_{i,t}+c_1(\mathbf{w}_{i,t}^p-\mathbf{w}_{i,t})+c_2 (\mathbf{w}_t^g-\mathbf{w}_{i,t})}_{\mathbf{BI}}{\color{black}\underbrace{-\alpha\nabla F_i(\mathbf{w}_{i,t};\mathfrak{D}_i\cup\mathfrak{D}^G_{tr})}_{\mathbf{AI}}}\, ,
\end{equation}
where $\mathfrak{D}^G_{tr}$  is a part of the global dataset $\mathfrak{D}^G$ and used for training to relieve the non-i.i.d. problem.
Thanks to the combination of the gradient-free stochastic optimization of the BI term and the gradient-based learning technique of the AI term in \eqref{eq:DSL_w}, the workers are good at searching for the optimal solutions to complex problems with fast convergence.

Then, the local workers calculate their own historical minimum loss function values and maintain their own historical best model parameters
\begin{equation}
  \label{eq:DSL_wp}
 \{F_{i,t+1}^p,\mathbf{w}^p_{i,t+1}\} = \arg\!\!\!\!\!\!\min_{\tau=1,\cdots,t+1} F_i(\mathbf{w}_{i,\tau}, \mathfrak{D}^G_{sc}),
\end{equation}
where $\mathfrak{D}^G_{sc}$ is the other part of the global dataset $\mathfrak{D}^G$ and {\color{black}used to provide fair-value scores of local models for best-worker selection by assessing}
the per-worker $F_{i,t+1}^p$ that helps to accurately identify $\mathbf{w}^p_{i,t+1}$. 
Then, all workers report their $F_{i,t+1}^p$ to the PS.

Comparing the received $\{F_{i,t+1}^p\}_i$ from all local workers, the PS selects the best worker $i_{t+1}^\star$ with the global optimum function value
\begin{equation}\label{eq:DSL_Fg}
 \{i_{t+1}^\star, F_{t+1}^g\} = \arg\!\!\!\!\min_{i=1,\cdots,U} F^p_{i,t+1}.
\end{equation}
If $F_{t+1}^g < F_{t}^g$, then the worker with the selected index $i_{t+1}^\star$ is invited to upload its $\mathbf{w}_{i_{t+1}^\star, t+1}^p$ to the PS as the globally best model parameter $\mathbf{w}^g_{t+1}= \mathbf{w}_{i_{t+1}^\star, t+1}^p$.
If $F_{t+1}^g \geq F_{t}^g$, then no worker is invited to upload local model parameter and the PS simply maintains the globally best model parameter and the globally best loss function value from the previous iteration as  $\mathbf{w}^g_{t+1}= \mathbf{w}^g_{t}$ and $F_{t+1}^g = F_{t}^g$.

Upon receiving $\mathbf{w}_{i_{t+1}^\star, t+1}^p$ from the invited worker, the PS further uses $\mathfrak{D}^G_{sc}$ to verify the reported model parameter.
If $F(\mathbf{w}_{i_{t+1}^\star, t+1}^p, \mathfrak{D}^G_{sc}) \neq F_{t+1}^g$, then a Byzantine attack is identified and the attacker is {\color{black}filtered out}; 
the PS will inquire the next best local worker, until confirmed.

{\em Communication Efficiency.}
Note that our CB-DSL requires $U$ workers to share their  function value $F^p_{i,t+1}$ which is only a scalar, and then invites only one local worker with the {\color{black}global minimum 
loss function value calculated using $\mathfrak{D}^G_{sc}$} to report its model parameter to the PS. Thus, our CB-DSL can dramatically reduce the overall communication overhead and energy consumption in edge networks during each communication round, compared with that required by standard FL approaches.

{\em Byzantine Robustness.}
In the process of collecting $F^p_{i,t+1}$'s from local workers, it is inherently vulnerable to Byzantine attacks. For example, a malicious worker may send a fake {\color{black}$\bar{F}^p_{i,t+1}$ $(< F^p_{i,t+1})$} to fool the PS to invite {\color{black}the attacker} 
to upload its fake model parameter as the global optimum, which will undermine edge learning.
Thanks to the introduction of $\mathfrak{D}^G_{sc}$ in our CB-DSL, it enables the PS to screen and remove the potential Byzantine attackers, {\color{black}resulting 
our CB-DSL 
Byzantine robust}.
%
\begin{algorithm}[!htb]
	\caption{CB-DSL}
	\label{alg:policyforCBFedPSO}
	\begin{algorithmic}[1]
\renewcommand{\algorithmicrequire}{\textbf{Initialization:}}
		\REQUIRE ~~\\
		$\mathbf{w}^p_{i,0}=\mathbf{w}_{i,0}$, $F^p_{i,0}=F_i(\mathbf{w}_{i,0}, \mathfrak{D}_{sc}^G)$, $\forall i$; 
\\
    \FOR {each iteration $t=1:T$}\vspace{0.08in}
    \STATE \!\!\!\!\textbf{at the local workers:}
        \STATE \hspace{0.1in} update the local model parameter $\mathbf{w}_{i,t+1}$ via \eqref{eq:DSL_w}; 
        \STATE  \hspace{0.1in}  calculate the per-worker historical minimum loss function value $F_{i,t+1}^p$ and maintain the corresponding per-worker historical best model parameter $\mathbf{w}_{i,t+1}^p$ via \eqref{eq:DSL_wp};
        \STATE  \hspace{0.1in}  send the scalar function value $F_{i,t+1}^p$ to the PS;
        \STATE   \hspace{0.1in} only the invited local worker sends its $\mathbf{w}^p_{i,t+1}$ to the PS;
    \vspace{0.1in}
    \STATE \!\!\!\!\textbf{at the PS:}
    \STATE \hspace{0.1in} {\color{black}compare} 
    the received $F^p_{i,t+1}$'s, select the best worker $i_{t+1}^\star$ and identify its function value as $F_{t+1}^g$ via \eqref{eq:DSL_Fg};
    \STATE \hspace{0.1in} {\em if} $F_{t+1}^g < F_{t}^g$, {\em then} invite the selected worker $i_{t+1}^\star$ to upload its model parameter as the globally best model parameter  $\mathbf{w}^g_{t+1}= \mathbf{w}_{i_{t+1}^\star, t+1}^p$;
    \STATE \hspace{0.1in} {\em else}, no worker is invited and maintain the globally best model parameter and function value {\color{black}from  the previous iteration} as $\mathbf{w}^g_{t+1}= \mathbf{w}^g_{t}$ and $F_{t+1}^g = F_{t}^g$;
%
    \STATE \hspace{0.1in}  given {\color{black} $\mathbf{w}_{i_{t+1}^\star, t+1}^p$ received from the invited worker}, verify $F(\mathbf{w}_{i_{t+1}^\star, t+1}^p, \mathfrak{D}^G_{sc})  == F_{t+1}^g$;
    \STATE \hspace{0.1in} {\em if} an attacker is identified by $F(\mathbf{w}_{i_{t+1}^\star, t+1}^p, \mathfrak{D}^G_{sc}) \neq F_{t+1}^g$, remove it and repeat line 8 until a legitimate worker is selected.
    \ENDFOR
	\end{algorithmic}
\end{algorithm}

\section{Convergence Analysis}\label{sec:Convergence}
In this section, 
we first make some definitions and assumptions for convergence analysis. With these preliminaries, the convergence behavior of our CB-DSL approach is theoretically evaluated {\color{black}by deriving 
an upper bound of the convergence rate.} 

\subsection{Assumption and Definition}
\begin{assumption}\label{ass1}
(Lipschitz continuity, smoothness): 
{\color{black}The gradient $\nabla F_i(\mathbf{w})$ of the loss function $F_i(\mathbf{w})$ at node $i$} is uniformly Lipschitz continuous with respect to $\mathbf{w}$, that is,
\begin{eqnarray}\label{eq:Lipschitz}
\|\nabla F_i(\mathbf{w}_{i,t+1})-\nabla F_i(\mathbf{w}_{i,t})\|\leq L\|\mathbf{w}_{i,t+1}-\mathbf{w}_{i,t}\|, \;\;\; \forall i, \mathbf{w}_{i,t}, \mathbf{w}_{i,t+1},
\end{eqnarray}
where $L$ is a positive constant, referred as the Lipschitz constant for the loss function $F_i(\cdot)$~\cite{chen2020joint}.
\end{assumption}

{\color{black}To facilitate analyses, we first rewrite $\mathbf{w}_{i,t}^p$ and $\mathbf{w}_t^g$ in \eqref{eq:DSL_w} as}
\begin{align}
\mathbf{w}_{i,t}^p&=\mathbf{w}_{i,t-1} {\color{black}+ \mathbf{v}_{i,t}^p,}\\
\mathbf{w}_t^g&=\mathbf{w}_{i,t-1} {\color{black}+ \mathbf{v}_{t}^g,}
\end{align}
where $\mathbf{v}_{i,t}^p$ and $\mathbf{v}_{t}^g$ denote the per-worker and globally optimal velocities currently used at the $i$-th worker. 

{\color{black}Then, the DSL velocity update $\mathbf{v}_{i,t+1} = \mathbf{BI}+\mathbf{AI} =  \mathbf{w}_{i,t+1} - \mathbf{w}_{i,t}$} in \eqref{eq:DSL_w} can be rewritten as
{\color{black}
\begin{align}\label{eq:DSL_velocity}
\mathbf{v}_{i,t+1} &=  c_0 \mathbf{v}_{i,t}+c_1(\mathbf{v}_{i,t}^p - (\mathbf{w}_{i,t}-\mathbf{w}_{i,t-1}) )+c_2 (\mathbf{v}_{t}^g - (\mathbf{w}_{i,t} - \mathbf{w}_{i,t-1})) - \alpha\nabla F_i(\mathbf{w}_{i,t})\nonumber\\
&=c_0 \mathbf{v}_{i,t}+c_1(\mathbf{v}_{i,t}^p-\mathbf{v}_{i,t})+c_2 (\mathbf{v}_{t}^g-\mathbf{v}_{i,t})-\alpha\nabla F_i(\mathbf{w}_{i,t})\nonumber\\
&=(c_0-c_1-c_2) \mathbf{v}_{i,t}+c_1\mathbf{v}_{i,t}^p+c_2 \mathbf{v}_{t}^g-\alpha\nabla F_i(\mathbf{w}_{i,t}),
\end{align}
}
where we replace $\nabla F_i(\mathbf{w}_{i,t};\mathfrak{D}_i\cup\mathfrak{D}^G_{tr})$ by $\nabla F_i(\mathbf{w}_{i,t})$ hereafter for symbol simplicity.

We use $\theta_{i,t}$, $\theta^p_{i,t}$, and $\theta^g_{t}$ to denote the angles {\color{black}between $\mathbf{v}_{i,t}$ and $-\nabla F_i(\mathbf{w}_{i,t})$, between $\mathbf{v}^p_{i,t}$ and $-\nabla F_i(\mathbf{w}_{i,t})$, and between $\mathbf{v}^g_{t}$ and $-\nabla F_i(\mathbf{w}_{i,t})$,} for any $i$ and $t$, respectively. Then we have
\begin{align}
\cos \theta_{i,t} &\triangleq \frac{{\color{black}\langle\mathbf{v}_{i,t}, -\nabla F_i(\mathbf{w}_{i,t})\rangle}}{\|\mathbf{v}_{i,t}\|\|\nabla F_i(\mathbf{w}_{i,t})\|}, \ \forall i, t,\\
\cos \theta^p_{i,t}&\triangleq \frac{{\color{black}\langle\mathbf{v}^p_{i,t}, -\nabla F_i(\mathbf{w}_{i,t})\rangle}}{\|\mathbf{v}^p_{i,t}\|\|\nabla F_i(\mathbf{w}_{i,t})\|}, \ \forall i, t,\\
\cos \theta^g_{t} &\triangleq \frac{{\color{black}\langle\mathbf{v}^g_{t}, -\nabla F_i(\mathbf{w}_{i,t})\rangle}}{\|\mathbf{v}^g_{t}\|\|\nabla F_i(\mathbf{w}_{i,t})\|}, \ \forall i, t.
\end{align}
We {\color{black}further} assume that the above cosine-similarity measures are bounded, whose lower and upper bounds are denoted as
\begin{align}\label{eq:desitascope}
\underline{q} &\leq  \cos \theta_{i,t}\leq \overline{q},\ \forall i,t \\
\underline{q}^p&\leq  \cos \theta^p_{i,t}\leq \overline{q}^p,\ \forall i, t\\
\underline{q}^g &\leq  \cos \theta^g_{t}\leq \overline{q}^g,\ \forall i,t,
\end{align}
%
\begin{align}\label{eq:denormscope1}
\underline{u} \leq& \frac{\|\mathbf{v}_{i,t}\|}{\|\nabla F_i(\mathbf{w}_{i,t})\|}\leq\overline{u}, \ \forall i, t, \\
\underline{u}^p\leq& \frac{\|\mathbf{v}^p_{i,t}\|}{\|\nabla F_i(\mathbf{w}_{i,t})\|}\leq \overline{u}^p,\ \forall i,t, \\
\underline{u}^g \leq& \frac{\|\mathbf{v}^g_{t}\|}{\|\nabla F_i(\mathbf{w}_{i,t})\|}\leq\overline{u}^g,\ \forall i,t.\label{eq:denormscope}
\end{align}


\subsection{Convergence Bound}
We adopt the expected improvement {\color{black}on the gradient in terms of its $\ell2$ norm, working as} 
an indicator of convergence for non-convex optimization\cite{wang2021cooperative,bernstein2018signsgd}
\begin{align}\label{eq:convergenceshow}
\min_{0,1,...,T}\mathbb{E}[\|\mathbf{g}_t\|^2]\leq\mathbb{E}\left[\sum_{t=1}^{T}\frac{1}{T}\|\mathbf{g}_t\|^2\right], 
\end{align}
where the norm of the gradient is expected to converge to 0 as $T$ increases to infinity, which means that the solution converges asymptotically.

With the assumptions and definitions presented in Subsection IV.A,
the convergence errors of the CB-DSL algorithm developed in Subsection III.C are bounded by the following \textbf{Theorem~\ref{theorem1}}. 

\begin{theorem}\label{theorem1}
For $T$ communication rounds, the expected convergence rate at each local worker in CB-DSL is bounded by
\begin{align}\label{eq:theorem1}
\mathbb{E}\left[\sum_{t=1}^{T}\frac{\|\nabla F_i(\mathbf{w}_{i,t})\|^2}{T}\right]\leq \frac{F(\mathbf{w}_{i,0})-F(\mathbf{w}^*)}{T\Phi_{E}},\ \forall i
\end{align}
where $\Phi_{E}={\color{black}\alpha-\frac{2c_0-\delta_{c_1}-\delta_{c_2}}{2}}\underline{q}\underline{u}
-\frac{\delta_{c_1}}{2} \overline{u}^p\overline{q}^p-\frac{\delta_{c_2}}{2} \overline{u}^g\overline{q}^g
-2L((c_0^2{\color{black}-\delta_{c_1}c_0-\delta_{c_2}c_0}+\frac{\delta_{c_1}^2}{3}
+\frac{\delta_{c_2}^2}{3} +\frac{\delta_{c_1}\delta_{c_2}}{2})\overline{u}^2
+\frac{\delta_{c_1}^2}{3}(\overline{u}^p)^2+\frac{\delta_{c_2}^2}{3} (\overline{u}^g)^2+\alpha^2)$.
\end{theorem}
\begin{proof}
Please refer to Appendix \ref{Appendix_A}.
\end{proof}

The result of \textbf{Theorem~\ref{theorem1}} implies the following 
convergence rate
\begin{align}\label{eq:conver1}
\mathbb{E}\left[\sum_{t=1}^{T}\frac{\|\nabla F_i(\mathbf{w}_{i,t})\|^2}{T}\right]\leq \mathcal{O}(\frac{1}{T\Phi_E}).
\end{align}
The inequality of \eqref{eq:conver1} indicates that the convergence of the CB-DSL is guaranteed as the number of communication rounds $T$ goes large. \textcolor[rgb]{0.00,0.00,0.00}{That is, as $T\rightarrow \infty$, we have $\mathbb{E}\left[\sum_{t=1}^{T}\frac{\|\nabla F_i(\mathbf{w}_{i,t})\|^2}{T}\right] \rightarrow 0$.}

\begin{remark}
When $c_0$, $\delta_{c_1}$, and $\delta_{c_2}$ are all set to be $0$, we have $\Phi_E=\alpha-2L\alpha^2$ in \eqref{eq:theorem1} and \eqref{eq:conver1},
and CB-DSL degenerates into FedAvg.
As {\color{black}$\Phi_E-(\alpha-2L\alpha^2)=\frac{\delta_{c_1}+\delta_{c_2}-2c_0}{2}\underline{q}\underline{u} +2L((\delta_{c_1}c_0+\delta_{c_2}c_0 - c_0^2 - \frac{\delta_{c_1}^2}{3}
-\frac{\delta_{c_2}^2}{3} -\frac{\delta_{c_1}\delta_{c_2}}{2})\overline{u}^2
-\frac{\delta_{c_1}^2}{3}(\overline{u}^p)^2-\frac{\delta_{c_2}^2}{3} (\overline{u}^g)^2) -\frac{\delta_{c_1}}{2} \overline{u}^p\overline{q}^p -\frac{\delta_{c_2}}{2} \overline{u}^g\overline{q}^g > 0$}, CB-DSL converges faster than FedAvg.
\end{remark}


\section{Model Divergence Analysis for the Case of Non-i.i.d. Data}\label{sec:weightdiver}
Intuitively, when the local datasets $\mathfrak{D}_i$ over different local workers are non-i.i.d., the learning performance varies with the degree of the local dataset heterogeneity.
Specifically, the greater the heterogeneity of local datasets, the model parameters updated at different local workers will become more diverse, e.g., with a larger range of the values of $\cos \theta_{i,t}$  in \eqref{eq:desitascope} among workers. That is, {\color{black}$\underline{q}$ and $\overline{q}$ in \eqref{eq:desitascope} go smaller and bigger, respectively. As a result, $\Phi_E$ defined in \eqref{eq:theorem1} decreases} 
as the heterogeneity of non-i.i.d. datasets increases, which {\color{black}leads to a loose upper bound on the convergence guarantee} 
in \eqref{eq:theorem1} and \eqref{eq:conver1} {\color{black}and thus degrades the learning performance with distributed non-i.i.d. datasets}. 
In this section, we provide {\color{black}a} statistical analysis to evaluate the impact of the local data heterogeneity on the learning performance of the CB-DSL.
We study the model parameter divergence resulted from the distance enlargement between the non-i.i.d. data distributions on local workers and the overall population distribution. 
We {\color{black}evaluate} such a distance by measuring the earth mover's distance (EMD) between these distributions \cite{rubner2000earth,zhao2018federated}. 

Consider a $C$-class classification problem defined over a compact space $\mathcal{X}$ and a label space $\mathcal{Y}$. The $k$-th data point $(\mathbf{x}_{i,k},y_{i,k})$ on the $i$-th local worker distributes over $\mathcal{X}\times \mathcal{Y}$ following the distribution $p_i$.
For the purpose of model divergence analysis, 
{\color{black}suppose} a genie worker who has the population data that reflect the population distribution $p$ of all local workers that may differ from $p_i$. The genie worker uses such knowledge of $p$ to 
{\color{black}search} for the globally optimal solution to the learning model, which serves as the {\color{black}reference} 
to calibrate the model divergence {\color{black}due to} 
the distributed non-i.i.d. data.
Then the original population loss function $F(\mathbf{w}):=\mathbb{E}_{\mathfrak{D}}[f(\mathbf{w};\mathbf{x}_{i,k},y_{i,k})]$ can be rewritten as 
\begin{align}
F(\mathbf{w})= \mathbb{E}_{\mathbf{x},y\sim p}\left[ \sum_{c=1}^C \mathds{1}_{y=c} f_c(\mathbf{x},\mathbf{w}) \right] = \sum_{c=1}^Cp(y=c)\mathbb{E}_{\mathbf{x}|y=c}[f_c(\mathbf{x},\mathbf{w})],
\end{align}
where $f_c$ denotes the probability for the $c$-th class, $c \in \{1,C\}$.

Then, the learning problem at the genie worker 
can be formulated as
\begin{align}
 \textbf{P2:} \quad \mathbf{w}^*= \arg \min_{\mathbf{w}}& \quad \sum_{c=1}^Cp(y=c)\mathbb{E}_{\mathbf{x}|y=c}[f_c(\mathbf{x},\mathbf{w})].
\end{align}
By solving \textbf{P2}, the model obtained at the genie worker {\color{black}plays as} 
the globally optimal position in  each communication round of CB-DSL. Then according to \eqref{eq:DSL_velocity}, the velocity at the genie worker in the $(t+1)$-th communication round is updated via
{\color{black}\begin{align}\label{eq:spgen}
\mathbf{v}^g_{t+1}&=c_0 \mathbf{v}^g_{t}-\alpha\nabla F(\mathbf{w}^g_{t}).
\end{align}}
The model parameter at the genie worker in the $(t+1)$-th communication round is updated as
{\color{black}\begin{align}\label{eq:weigh}
\mathbf{w}^g_{t+1}&=\mathbf{w}^g_{t}+\mathbf{v}^g_{t+1}.
\end{align}}
Given \eqref{eq:DSL_w} and \eqref{eq:weigh}, the model divergence between the $i$-th local worker and the genie worker is defined as
\begin{align}
model\,\,\, divergence=\frac{\|\mathbf{w}_{i,t+1}-\mathbf{w}^g_{t+1}\|}{\|\mathbf{w}^g_{t+1}\|}.
\end{align}

Next, we provide \textbf{Theorem~\ref{theo:theorem2}} to evaluate the model divergence by deriving its upper bound theoretically.
\begin{theorem}\label{theo:theorem2}
Under the assumption that $\nabla\mathbb{E}_{\mathbf{x}|y=c}[f_c(\mathbf{x},\mathbf{w})]$ is $L_c$-Lipschitz for each class $c {\in} \{1,C\}$, we have the following inequality for the model divergence after $(t+1)$ communication rounds
\begin{align}
\|\mathbf{w}_{i,t+1}-\mathbf{w}^g_{t+1}\|
&\leq \beta^{t+1}\|\mathbf{w}_{i,0}
-\mathbf{w}^g_{0}\|+{\color{black}|c_0-c_1-c_2|}\sum_{j=0}^t\beta^{t-j}\| \mathbf{v}_{i,j}-\mathbf{v}_{j}^g\|
\nonumber\\
&
+\alpha \sum_{c=1}^C\|p_i(y=c)-p(y=c)\|\sum_{j=0}^tf_{max}(\mathbf{w}^g_{j}),\label{eq:theo2}
\end{align}
where $\beta=1+\alpha \sum_{c=1}^Cp_i(y=c)L_c$ and $f_{max}{\color{black}(\mathbf{w}^{g}_{j})}= \max\{\nabla\mathbb{E}_{\mathbf{x}|y=c}[f_c(\mathbf{x},\mathbf{w}^g_{j})]\}_{c=1}^C$.
\end{theorem}
\begin{proof}
Please refer to Appendix \ref{Appendix_B}.
\end{proof}

\begin{remark}
Our theoretical result of \textbf{Theorem~\ref{theo:theorem2}} indicates that the model divergence can be upper bounded in \eqref{eq:theo2} after $(t+1)$ communication rounds, which mainly comes from three parts, including the initial model divergence, i.e., $\|\mathbf{w}_{i,0}-\mathbf{w}^g_{0}\|$, the velocity divergence after $t$ communication rounds, i.e., $\| \mathbf{v}_{i,j}-\mathbf{v}_{j}^g\|$, and the model divergence induced by the probability distance between the data distribution on the $i$-th local worker and the ground truth distribution for the whole population as on the genie worker, i.e., $\sum_{c=1}^C\|p_i(y=c)-p(y=c)\|$.
\end{remark}
\begin{remark}
In \eqref{eq:theo2}, the initial model divergence (first term) and the velocity divergence (second term) after $(t+1)$ communication rounds are iteratively amplified by $\beta$. Since $\beta>1$, if different local workers start from different initial model parameters in the standard FL, then the model divergence will still be enlarged, even though the local workers have  i.i.d. data.
\end{remark}
\begin{remark}
In \eqref{eq:theo2}, the third term $\sum_{c=1}^C\|p_i(y=c)-p(y=c)\|$ is the EMD between the data distribution on the $i$-th local worker and the population distribution \cite{rubner2000earth}, when the distance metric is defined as $\|p_i(y=c)-p(y=c)\|$. The impact of EMD is affected by the learning rate $\alpha$, the number of communication rounds $t$, and the class-wise maximum gradient $f_{max}(\mathbf{w}_{j})$.
\end{remark}

\section{Experimental Results}
This section demonstrates that our CB-DSL with a small amount of globally shared dataset outperforms the benchmark methods, with  better learning performance and  faster convergence speed, on both the i.i.d. and non-i.i.d. settings, even in the presence of Byzantine attacks. 
\subsection{System and Dataset Setting}
To evaluate the learning performance of our CB-DSL, 
we perform empirical simulations by conducting a handwritten-digit classification task based on the widely-used MNIST dataset\footnote{http://yann.lecun.com/exdb/mnist/} that consists of 10 classes ranging from digit ``0" to ``9". In the MNIST dataset, there are 60000 training samples and 10000 testing samples.
In the training procedure, we set the total number of local workers to be $U=50$, as the IoT devices in an edge network.
For each local worker in the i.i.d. setting, 300 distinct training samples are randomly selected as its local datasets, i.e., $K_i =300, \forall i$.
To build the non-i.i.d. data setting upon the MNIST dataset,
we first sort all the 60000 training samples based on the classification labels. Then we divide the 60000 training samples into 200 shards, each of which consists 300 samples, that are highly non-i.i.d. shard by shard \cite{mcmahan2017communication}. We randomly allocate two shards to each local worker for the edge learning problem.
The globally shared scoring dataset $\mathfrak{D}^G_{sc}$ consists of 2000 data samples, 
and the globally shared training dataset $\mathfrak{D}^G_{tr}$ consists of 150 data samples for the i.i.d. setting and 600 data samples for the non-i.i.d. setting.
In addition, we set 
$c_0=1$, $\delta_{c_1}=1$, and $\delta_{c_2}=1$.
\subsection{Neural Network Setting}
For the learning model architecture, we use a five-layer Convolutional Neural
Network (CNN) whose detailed hyperparameter settings are listed in {Table \ref{table:1}}.
\begin{table*}[!tb]
 \caption{Model architecture of the experiment.}\label{table:1}
   \centering
 \begin{tabular}{|c|c|}
  \hline
  Layer & Details \\
  \hline
  1 & \makecell*[c]{Conv2D(1, 6, 5)\\
                     ReLU, MaxPool2D(2, 2)}  \\
  \hline
  2 & \makecell*[c]{Conv2D(6, 16, 5)\\
                     ReLU, MaxPool2D(2, 2)} \\
  \hline
  3 & \makecell*[c]{FC(16 * 4 * 4, 120)\\
                     ReLU} \\
  \hline
  4 & \makecell*[c]{FC(120, 84)\\
                     ReLU} \\
  \hline
  5 & \makecell*[c]{FC(84,10)} \\
  \hline
\end{tabular}
\end{table*}
%
For the convolutional layers (Conv2D), we list the sizes of the parameters with sequence of input and output dimensions, and kernel size. For the max pooling layers (MaxPool2D), we list kernel and stride sizes. For the fully-connected layers (FC), we list input and output dimensions.
During the training process, we use the SGD optimizer with learning rate $\alpha=0.005$ and the cross-entropy loss. The batch size is set as $\|\mathfrak{B}_i\|=10, \forall i$ for the mini-batch SGD\cite{mcmahan2017communication,qian2015efficient}.
\subsection{Different Approaches}
{\color{black}We compare the proposed CB-DSL with FedAvg~\cite{konevcny2016federated}, given either i.i.d. or non-i.i.d. data, for different cases of globally shared dataset}
(without any shared dataset, with shared dataset for scoring, with shared dataset for training, with shared dataset for both scoring and training), including:
\begin{enumerate}
  \item \emph{FedAvg without any globally shared dataset $\mathfrak{D}^G$}: it is the standard FedAvg \cite{konevcny2016federated}.
  \item \emph{\textcolor[rgb]{0.00,0.00,0.00}{CB-DSL without any globally shared dataset $\mathfrak{D}^G$}}: the local workers use their own local dataset to calculate $F^p_{i,t}$.
  \item \emph{CB-DSL with a globally shared dataset for scoring $\mathfrak{D}^G_{sc}$}: the local workers use the globally shared scoring dataset to calculate $F^p_{i,t}$ in CB-DSL.
  \item \emph{FedAvg with a globally shared dataset for training $\mathfrak{D}^G_{tr}$}: the local workers use both their own local dataset and the globally shared training dataset to train their local models in standard FedAvg\cite{konevcny2016federated}.
  \item \emph{CB-DSL with a globally shared dataset for both training $\mathfrak{D}^G_{tr}$ and scoring $\mathfrak{D}^G_{sc}$}: the local workers use both their own local dataset and the globally shared training dataset to train their local models and then use the globally shared scoring dataset to calculate $F^p_{i,t}$ in CB-DSL.
\end{enumerate}
\subsection{Evaluation and Comparison}
\begin{figure}[tb]
  \centering
  \includegraphics[scale=0.55]{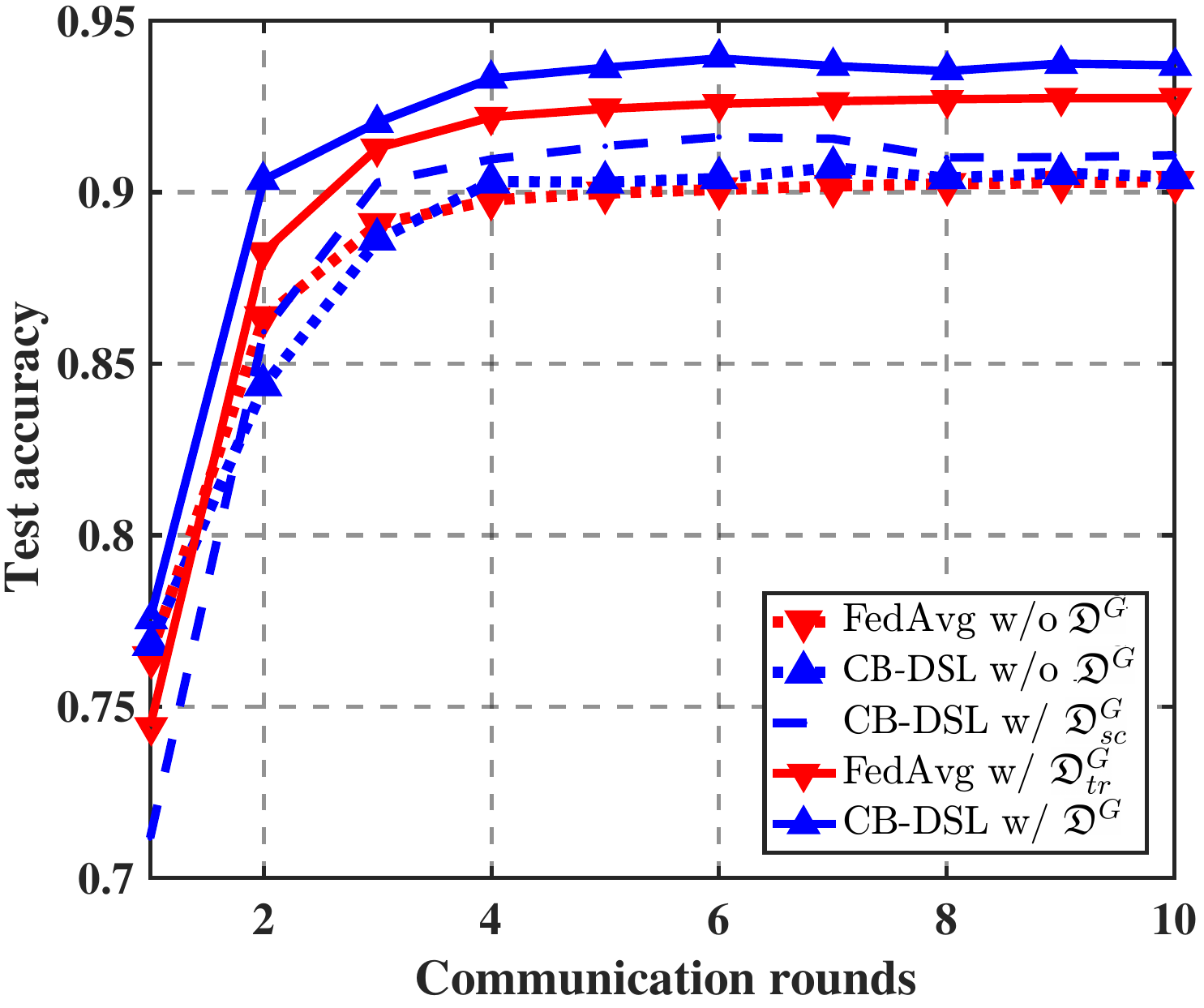}
    \caption{The performance comparison under the i.i.d. setting.}\label{fig:iid}
\end{figure}

\begin{figure}[tb]
  \centering
  \includegraphics[scale=0.55]{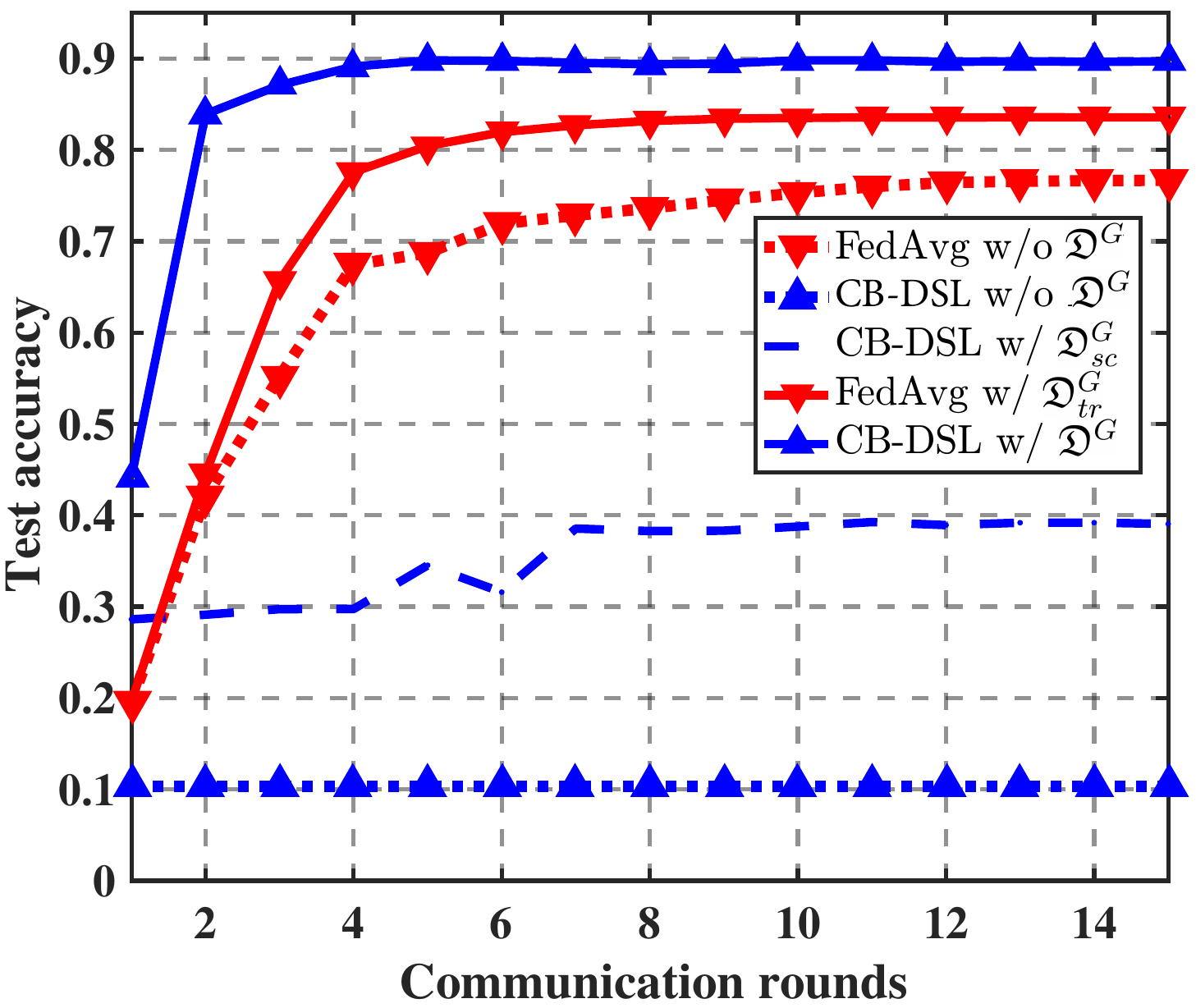}
    \caption{The performance comparison under the non-i.i.d. setting.}\label{fig:noniid}
\end{figure}

Fig. \ref{fig:iid} and Fig. \ref{fig:noniid} show the simulation results for the five cases under the i.i.d. and the non-i.i.d. settings, respectively. As shown in Fig. \ref{fig:iid}, \textcolor[rgb]{0.00,0.00,0.00}{CB-DSL without $\mathfrak{D}^G$} is slightly better than FedAvg under the same learning settings for the i.i.d. case.
A globally shared scoring dataset $\mathfrak{D}^G_{sc}$ introduced in CB-DSL can improve the learning performance of {\color{black}CB-DSL without any globally shared dataset.} 
This is because $\mathfrak{D}^G_{sc}$ can help to select the global optimum more accurately than that based on local workers simply using their own dataset which however makes the loss function $F(\cdot)$ only partially observable at local workers.
In addition, a globally shared training dataset $\mathfrak{D}^G_{tr}$ can further improve the learning performance of FedAvg and CB-DSL, since the data samples are increased for training. Meanwhile, CB-DSL is superior thanks to its benefits by leveraging the exploration-exploitation gains from the BI component and the fast convergence characteristics from the AI component.

\begin{figure}[tb]
  \centering
  \includegraphics[scale=0.55]{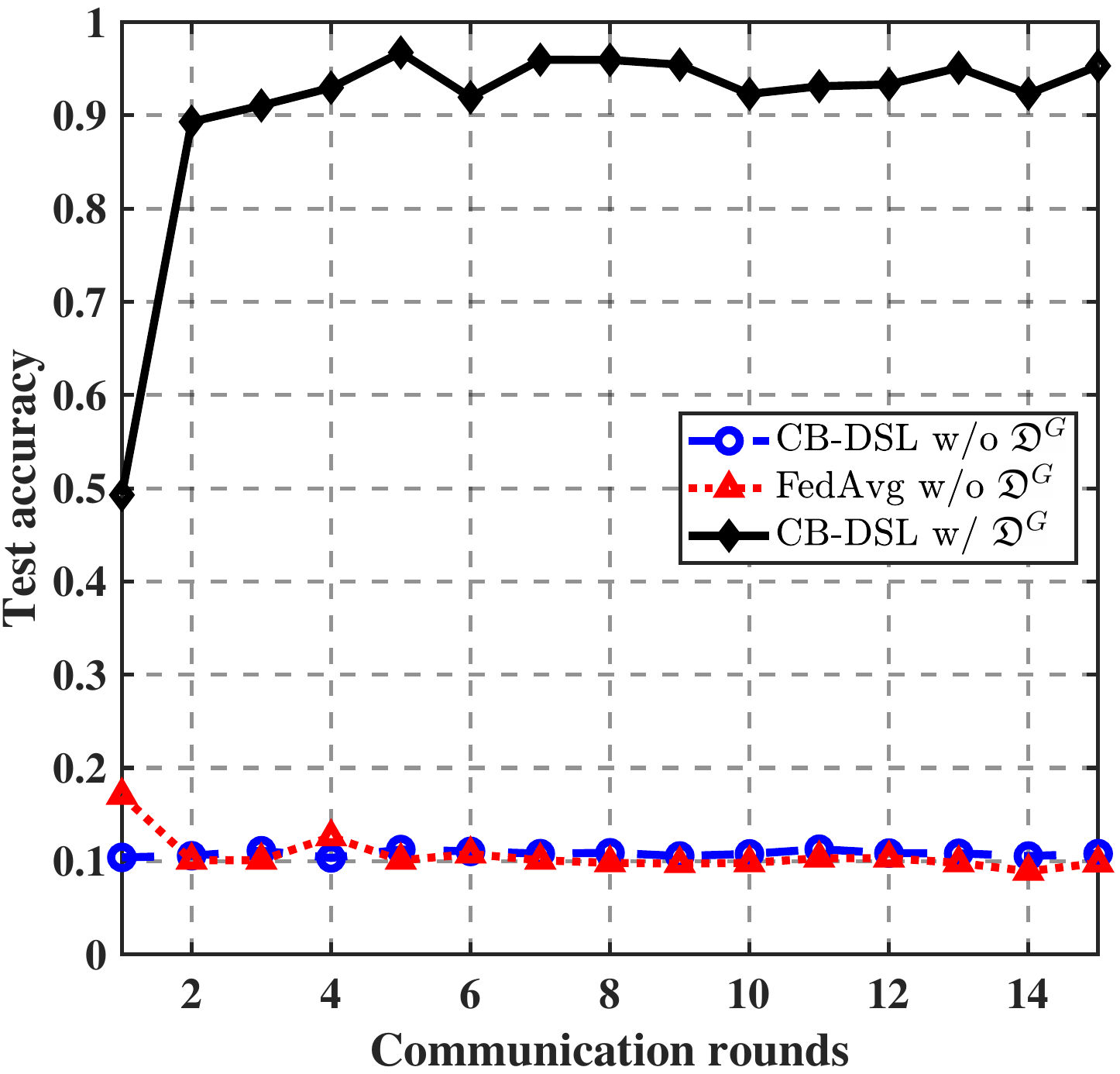}
    \caption{The performance comparison with a Byzantine attacker under the i.i.d. setting.}\label{fig:iidatt}
\end{figure}

\begin{figure}[tb]
  \centering
  \includegraphics[scale=0.55]{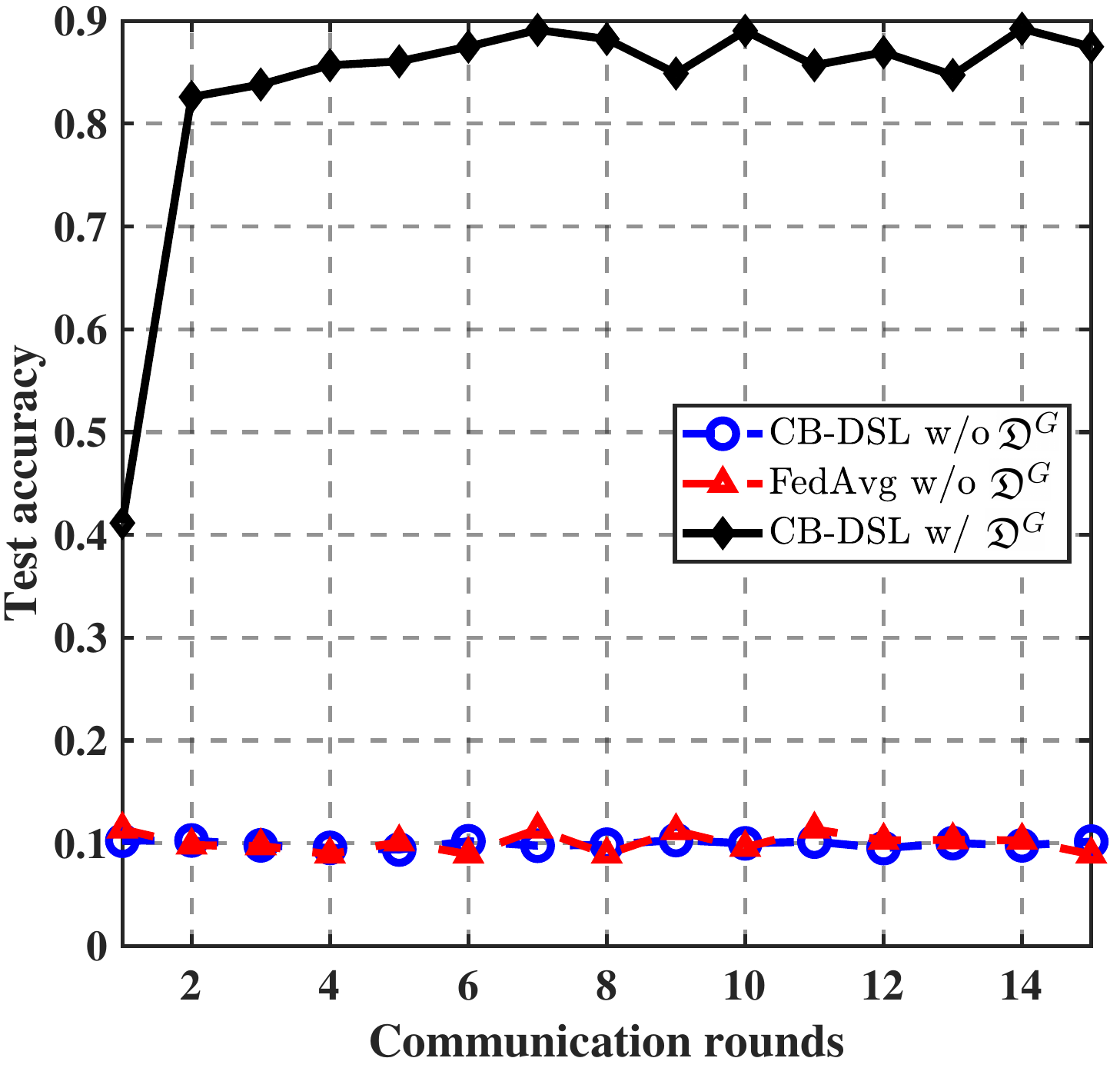}
    \caption{The performance comparison with a Byzantine attacker under the non-i.i.d. setting.}\label{fig:noniidatt}
\end{figure}
In Fig. \ref{fig:noniid}, \textcolor[rgb]{0.00,0.00,0.00}{when CB-DSL {\color{black}runs} without globally shared dataset for training $\mathfrak{D}^G_{tr}$, it cannot work properly in the non-i.i.d. setting.
This is because CB-DSL hinges} on single best worker selection which however may not hold the optimum model at all due to the model divergence from the ground truth population distribution point of view in the non-i.i.d. setting.
Although using a globally shared scoring dataset $\mathfrak{D}^G_{sc}$ can slightly improve the learning performance of CB-DSL, it is still worse than FedAvg {\color{black}where all workers with non-i.i.d. data contribute to model average at the cost of high communication cost}. 
When both a globally shared training dataset and scoring dataset are used as $\mathfrak{D}^G=\mathfrak{D}^G_{tr}\cup\mathfrak{D}^G_{sc}$, CB-DSL {\color{black}turns to outperform} 
FedAvg.
This is because $\mathfrak{D}^G_{tr}$ helps to relieve the local data heterogeneity issue by making the local datasets to become more i.i.d., which decreases the EMD between the data distributions on local workers and the population distribution as {\color{black}revealed by our model divergence analysis} in Section V. \textcolor[rgb]{0.00,0.00,0.00}{Besides, {\color{black}the improvement on learning accuracy also indicates that by} using the exploration-exploitation mechanism of PSO, our CB-DSL solutions have an increased chance to
jump out of local optimum traps via the swarm intelligence.}

In Fig. \ref{fig:iidatt} and Fig. \ref{fig:noniidatt}, we provide the performance comparison in the presence of the Byzantine attack for both the i.i.d. and the non-i.i.d. settings, respectively. It is obvious that even only one Byzantine attacker can {\color{black}fail FedAvg and} 
CB-DSL without $\mathfrak{D}^G$. On the other hand, the CB-DSL with $\mathfrak{D}^G$ can effectively defend the Byzantine attack, 
because the globally shared dataset for scoring $\mathfrak{D}^G_{sc}$ can {\color{black}help identify and} screen out the Byzantine attacker {\color{black}as explained in \textbf{Algorithm~\ref{alg:policyforCBFedPSO}}}.

\begin{figure}[tb]
  \centering
  \includegraphics[scale=0.55]{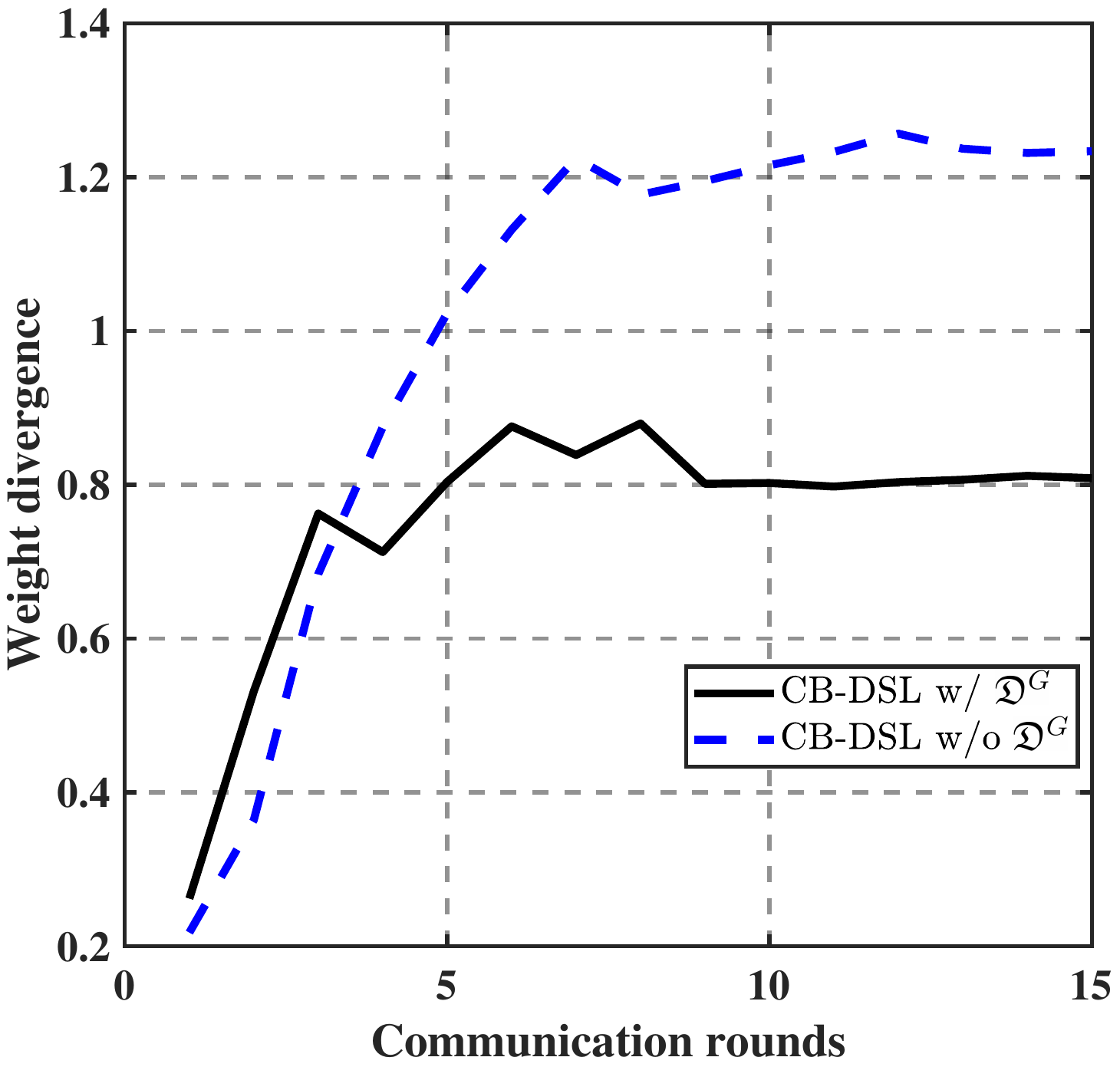}
    \caption{The comparison of the weight divergences under the non-i.i.d. setting.}\label{fig:weight_diver}
\end{figure}
\textcolor[rgb]{0.00,0.00,0.00}{In Fig.~\ref{fig:weight_diver}, we {\color{black}further evaluate 
the weight divergences effects} under the non-i.i.d. setting. As the communication rounds increase, the weight divergences of CB-DSL with or without $\mathfrak{D}^G$ {\color{black}first increase and then flatten out after several communication rounds.} 
The final steady-state weight divergence of the CB-DSL with $\mathfrak{D}^G$ is {\color{black}much less} than that of the CB-DSL without $\mathfrak{D}^G$, {\color{black}as depicted by the gap between the two curves in Fig.~\ref{fig:weight_diver}. Such a nontrivial gap confirms the theoretical results} 
of \textbf{Theorem~\ref{theo:theorem2}}: (1) the model divergence will be enlarged as the communication rounds increase (this is because that the initial model divergence is iteratively amplified by $\beta$, as explained in \emph{Remark 3}); (2) the use of global data $\mathfrak{D}^G$ can reduce the weight divergence (this is because that the use of $\mathfrak{D}^G$ decreases the EMD between the data distributions on local workers and the population distribution, as explained in \emph{Remark 4}).}

{\color{black}Note that} only one local worker is selected and {\color{black}invited} 
to send its model parameter to the PS in CB-DSL, while all workers need to send their model parameters to the PS in FedAvg.
Therefore, \textcolor[rgb]{0.00,0.00,0.00}{the communication cost consumed in CB-DSL is only $\frac{1}{U}$ of that in FedAvg,} given the fact that the communication cost for the transmission of loss function values as a scalar is relatively trivial to the transmission of the model parameter vector and thus can be ignored. \textcolor[rgb]{0.00,0.00,0.00}{In addition, we can see from Fig. \ref{fig:iid} and Fig. \ref{fig:noniid} that our CB-DSL with $\mathfrak{D}^G$ uses fewer communication rounds than FedAvg to achieve the same learning accuracy. {\color{black}As a result}, our CB-DSL is communication-efficient {\color{black}with less communication rounds and less communication overhead per round} in practical applications.}

%

\section{Conclusion}
This work studies a novel communication-efficient and Byzantine-robust distributed swarm learning (CB-DSL) approach for edge IoT systems, as a holistic integration of the AI-enabled SGD and the BI-enabled PSO. 
We propose to introduce a globally shared dataset to overcome the major challenging issues in edge learning including: the partially observability of loss function in distributed learning problems, the non-i.i.d. local data issues, and the potential Byzantine attacks.
We provide theoretical analysis of the convergence behavior of the proposed CB-DSL, which indicates that our method can achieve better learning performance than  existing distributed learning methods.
Further, we provide the model divergence evaluation of 
the proposed CB-DSL in the non-i.i.d. settings, 
which quantifies how a globally shared dataset can improve the learning performance of the CB-DSL in the non-i.i.d. setting.
Simulation results verify that our proposed CB-DSL solution can improve learning performance in both the i.i.d. and non-i.i.d. settings, \textcolor[rgb]{0.00,0.00,0.00}{compared with the standard FedAvg.} Meanwhile, the communication saving by the CB-DSL inherits the advantage of the bio-inspired PSO techniques with much reduced communication cost than standard FedAvg.

\section*{Acknowledgments}
This work was partly supported by the National Natural Science Foundation of China (Grants \#61871023 and \#61931001), Beijing Natural Science Foundation (Grant \#4202054), the National Science Foundation of the US (Grants
{\color{black}\#1939553, \#2003211,\#2128596, \#2136202 and \#2231209)}, and the Virginia Research Investment Fund (Commonwealth Cyber Initiative Grant \#223996).

\begin{appendices}
\section{Proof of \textbf{Theorem \ref{theorem1}}}\label{Appendix_A}
\begin{proof}
Because $F_i(\cdot)$ is $L$-smooth from Assumption 1, according to \cite[Lemma 3.4]{bubeck2015convex} {\color{black}and velocity update in \eqref{eq:DSL_velocity}}, we have
\begin{align}
F_i(\mathbf{w}_{i,t+1})-F_i(\mathbf{w}_{i,t})&\leq{\color{black}(\mathbf{w}_{i,t+1}-\mathbf{w}_{i,t})^T} \nabla F_i(\mathbf{w}_{i,t})
 +\frac{L}{2}\|\mathbf{w}_{i,t+1}-\mathbf{w}_{i,t}\|^2\nonumber\\
 &{\color{black}=\mathbf{v}_{i,t+1}^T}\nabla F_i(\mathbf{w}_{i,t})+\frac{L}{2}\|\mathbf{v}_{i,t+1}\|^2\nonumber\\
 &= {\color{black}(c_0-c_1-c_2) \mathbf{v}_{i,t}^T} \nabla F_i(\mathbf{w}_{i,t})+c_1{\color{black}(\mathbf{v}_{i,t}^p)^T}\nabla F_i(\mathbf{w}_{i,t})\nonumber\\
 &+c_2 {\color{black}(\mathbf{v}_{t}^g)^T}\nabla F_i(\mathbf{w}_{i,t})-\alpha\|\nabla F_i(\mathbf{w}_{i,t})\|^2+\frac{L}{2}\|\mathbf{v}_{i,t+1}\|^2
 .\label{eq:smooth0}
\end{align}

According to the definitions and assumptions of $\overline{q}$, $\overline{q}^p$, $\overline{q}^g$, $\underline{q}$, $\underline{q}^p$, $\underline{q}^g$, $\overline{u}$, $\overline{u}^p$, $\overline{u}^g$, $\underline{u}$, $\underline{u}^p$, $\underline{u}^g$ in \eqref{eq:desitascope}-\eqref{eq:denormscope}, for any $i$ and $t$, we have
\begin{align}
&\underline{u}\underline{q}\|\nabla F_i(\mathbf{w}_{i,t})\|^2\leq {\color{black}\mathbf{v}_{i,t}^T\nabla F_i(\mathbf{w}_{i,t})}=\|\mathbf{v}_{i,t}\|\|\nabla F_i(\mathbf{w}_{i,t})\|\cos \theta_{i,t}\leq \overline{u}\ \overline{q}\|\nabla F_i(\mathbf{w}_{i,t})\|^2,\label{eq:scale1}
\\
&\underline{q}^p\underline{u}^p\|\nabla F_i(\mathbf{w}_{i,t})\|^2\leq {\color{black}(\mathbf{v}^p_{i,t})^T \nabla F_i(\mathbf{w}_{i,t})}=\|\mathbf{v}^p_{i,t}\|\|\nabla F_i(\mathbf{w}_{i,t})\|\cos \theta^p_{i,t}\leq \overline{u}^p\overline{q}^p\|\nabla F_i(\mathbf{w}_{i,t})\|^2,\label{eq:scale2}
\\
&\underline{q}^g\underline{u}^g\|\nabla F_i(\mathbf{w}_{i,t})\|^2\leq {\color{black}(\mathbf{v}^g_{t})^T\nabla F_i(\mathbf{w}_{i,t})}=\|\mathbf{v}^g_{t}\|\|\nabla F_i(\mathbf{w}_{i,t})\|\cos \theta^g_{t}\leq \overline{u}^g\overline{q}^g\|\nabla F_i(\mathbf{w}_{i,t})\|^2.\label{eq:scale3}
\end{align}

Substituting \eqref{eq:scale1}-\eqref{eq:scale3} to \eqref{eq:smooth0}, we have
\begin{align}
F_i(\mathbf{w}_{i,t+1})&-F_i(\mathbf{w}_{i,t})
 \leq {\color{black}(c_0-c_1-c_2)}\underline{q}\underline{u}\|\nabla F_i(\mathbf{w}_{i,t})\|^2+c_1 \overline{u}^p\overline{q}^p\|\nabla F_i(\mathbf{w}_{i,t})\|^2
 \nonumber\\
 &+ c_2 \overline{u}^g\overline{q}^g\|\nabla F_i(\mathbf{w}_{i,t})\|^2- \alpha\|\nabla F_i(\mathbf{w}_{i,t})\|^2+\frac{L}{2}\|\mathbf{v}_{i,t+1}\|^2
 \nonumber\\
 &=(c_1 \overline{u}^p\overline{q}^p+c_2 \overline{u}^g\overline{q}^g{\color{black}+(c_0-c_1-c_2)}\underline{q}\underline{u}-\alpha)\|\nabla F_i(\mathbf{w}_{i,t})\|^2+\frac{L}{2}\|\mathbf{v}_{i,t+1}\|^2
 .
 \label{eq:smooth1}
\end{align}

Applying the triangle inequality of norms $\| \mathbf{X}+\mathbf{Y}\| \leq \| \mathbf{X}\|+\| \mathbf{Y}\|$, the submultiplicative property of norms $\| \mathbf{X}\mathbf{Y}\| \leq \| \mathbf{X}\|\| \mathbf{Y}\|$, and the Jensen’s inequality $(\sum_{i=1}^{n}a_i)^2\leq n\sum_{i=1}^{n}a_i^2$, we have
\begin{align}
\|\mathbf{v}_{i,t+1}\|^2&=\|{\color{black}(c_0-c_1-c_2) \mathbf{v}_{i,t}+c_1\mathbf{v}_{i,t}^p+c_2 \mathbf{v}_{t}^g-\alpha\nabla F_i(\mathbf{w}_{i,t})}\|^2
 \nonumber\\
 &\leq (\|{\color{black}(c_0-c_1-c_2)} \mathbf{v}_{i,t}\|+\|c_1\mathbf{v}_{i,t}^p\|+\|c_2 \mathbf{v}_{t}^g\|+\|\alpha\nabla F_i(\mathbf{w}_{i,t})\|)^2
 \nonumber\\
 &\leq 4({\color{black}(c_0-c_1-c_2)}^2\| \mathbf{v}_{i,t}\|^2+c_1^2\|\mathbf{v}_{i,t}^p\|^2+c^2_2\| \mathbf{v}_{t}^g\|^2+\alpha^2\|\nabla F_i(\mathbf{w}_{i,t})\|^2).
 \label{eq:jesen}
\end{align}

According to the assumptions of $\overline{u}$, $\overline{u}^p$, $\overline{u}^g$ in \eqref{eq:denormscope1}-\eqref{eq:denormscope}, for any $i$ and $t$, we have
\begin{align}
\|\mathbf{v}_{i,t}\|\leq \overline{u}\|\nabla F_i(\mathbf{w}_{i,t})\|,\label{eq:scale4}\\
\|\mathbf{v}^p_{i,t}\|\leq \overline{u}^p\|\nabla F_i(\mathbf{w}_{i,t})\|,\label{eq:scale5}\\
\|\mathbf{v}^g_{t}\|\leq \overline{u}^g\|\nabla F_i(\mathbf{w}_{i,t})\|.\label{eq:scale6}
\end{align}

Substituting \eqref{eq:scale4}-\eqref{eq:scale6} to \eqref{eq:jesen}, we have
\begin{align}
\|\mathbf{v}_{i,t+1}\|^2 &\leq 4({\color{black}(c_0\overline{u}-c_1\overline{u}-c_2\overline{u})}^2\|\nabla F_i(\mathbf{w}_{i,t})\|^2+c_1^2(\overline{u}^p)^2\|\nabla F_i(\mathbf{w}_{i,t})\|^2\nonumber\\
&+c^2_2(\overline{u}^g)^2\|\nabla F_i(\mathbf{w}_{i,t})\|^2+\alpha^2\|\nabla F_i(\mathbf{w}_{i,t})\|^2)\nonumber\\
&=4({\color{black}(c_0\overline{u}-c_1\overline{u}-c_2\overline{u})}^2+c_1^2(\overline{u}^p)^2+c^2_2(\overline{u}^g)^2+\alpha^2)\|\nabla F_i(\mathbf{w}_{i,t})\|^2.
 \label{eq:jesen1}
\end{align}

Substituting \eqref{eq:jesen1} to \eqref{eq:smooth1}, we have
\begin{align}
F_i(\mathbf{w}_{i,t+1})-F_i(\mathbf{w}_{i,t})
 &\leq( c_1 \overline{u}^p\overline{q}^p+c_2 \overline{u}^g\overline{q}^g {\color{black}+(c_0-c_1-c_2)}\underline{q}\underline{u}-\alpha)\|\nabla F_i(\mathbf{w}_{i,t})\|^2\nonumber\\
 &+2L({\color{black}(c_0\overline{u}-c_1\overline{u}-c_2\overline{u})}^2+c_1^2(\overline{u}^p)^2 +c^2_2(\overline{u}^g)^2+\alpha^2)\|\nabla F_i(\mathbf{w}_{i,t})\|^2\nonumber\\
 &=\Phi\|\nabla F_i(\mathbf{w}_{i,t})\|^2,
 \label{eq:smooth2}
\end{align}
where $\Phi= c_1 \overline{u}^p\overline{q}^p+c_2 \overline{u}^g\overline{q}^g{\color{black}+(c_0-c_1-c_2)}\underline{q}\underline{u} -\alpha+2L({\color{black}(c_0\overline{u}{-}c_1\overline{u}{-}c_2\overline{u})}^2+c_1^2(\overline{u}^p)^2 +c^2_2(\overline{u}^g)^2+\alpha^2)$.

Then we extend the expectation over randomness introduced by CB-DSL and mini-batch training data in the trajectory of iterations, and perform a telescoping sum  of \eqref{eq:smooth2} over the $T$ iterations
\begin{align}\label{sumexpectation}
F(\mathbf{w}_{i,0})-F(\mathbf{w}^*)&\geq F(\mathbf{w}_{i,0})-\mathbb{E}[F(\mathbf{w}_{i,T})]
\nonumber \\
&
=\mathbb{E}\left[\sum_{t=1}^{T}(F(\mathbf{w}_{i,t-1})-F(\mathbf{w}_{i,t}))\right]
\nonumber \\
&
\geq\mathbb{E}\left[ \sum_{t=1}^{T}\Phi_E\|\nabla F_i(\mathbf{w}_{i,t})\|^2\right],
\end{align}
where $\Phi_E=\mathbb{E}[-\Phi]=-\frac{\delta_{c_1}}{2} \overline{u}^p\overline{q}^p-\frac{\delta_{c_2}}{2} \overline{u}^g\overline{q}^g{\color{black}-\frac{2c_0-\delta_{c_1}-\delta_{c_2}}{2}}\underline{q}\underline{u}+\alpha -2L((c_0^2{\color{black}-\delta_{c_1}c_0-\delta_{c_2}c_0}+\frac{\delta_{c_1}^2}{3}+\frac{\delta_{c_2}^2}{3} +\frac{\delta_{c_1}\delta_{c_2}}{2})\overline{u}^2 +\frac{\delta_{c_1}^2}{3}(\overline{u}^p)^2+\frac{\delta_{c_2}^2}{3} (\overline{u}^g)^2+\alpha^2)$.

Finally, we can rearrange the inequality of \eqref{sumexpectation} to yield the convergence rate
\begin{align}\label{sumexpectation1}
\mathbb{E}\left[\sum_{t=1}^{T}\frac{\|\nabla F_i(\mathbf{w}_{i,t})\|^2}{T}\right]\leq \frac{F(\mathbf{w}_{i,0})-F(\mathbf{w}^*)}{T\Phi_E}.
\end{align}

Hence, the proof is completed.
%
\end{proof}

\section{Proof of \textbf{Theorem \ref{theo:theorem2}}}\label{Appendix_B}

\begin{proof}
Based on the definitions of $\mathbf{w}_{i,t+1}$ and $\mathbf{w}^g_{t+1}$ in \eqref{eq:DSL_w} and \eqref{eq:weigh}, we have
\begin{align}
\|\mathbf{w}_{i,t+1}-\mathbf{w}^g_{t+1}\|&={\color{black}\|\mathbf{w}_{i,t}-\mathbf{w}^g_{t}
+\mathbf{v}_{i,t+1}-\mathbf{v}^g_{t+1}}\|
\nonumber\\
&\leq \|{\color{black}\mathbf{w}_{i,t}}-\mathbf{w}^g_{t}\|+\|\mathbf{v}_{i,t+1}-\mathbf{v}^g_{t+1}\|.\label{eq:derwe}
\end{align}

Then based on the definitions of $\mathbf{v}_{i,t+1}$ and $\mathbf{v}^g_{t+1}$ in \eqref{eq:DSL_velocity} and \eqref{eq:spgen}, we get
\begin{align}
\|\mathbf{v}_{i,t+1}-\mathbf{v}^g_{t+1}\|
&=\|{\color{black}(c_0-c_1-c_2)} \mathbf{v}_{i,t}{\color{black}\,+\,c_1}\mathbf{v}_{i,t}^p {\color{black}\,+\,(c_2-c_0)} \mathbf{v}_{t}^g
{\color{black} \,-\, \alpha}\nabla F_i(\mathbf{w}_{i,t}){ \color{black}\, +\,\alpha}\nabla F(\mathbf{w}^g_{t})\|\nonumber\\
&\leq \|{\color{black}(c_0-c_1-c_2)} \mathbf{v}_{i,t}{\color{black}\,+\,c_1}\mathbf{v}_{i,t}^p{\color{black}\,+\,(c_2-c_0)} \mathbf{v}_{t}^g\|
+\|\alpha\nabla F_i(\mathbf{w}_{i,t})-\alpha\nabla F(\mathbf{w}^g_{t})\|\nonumber\\
&\leq \|{\color{black}(c_0-c_1-c_2)} \mathbf{v}_{i,t}{\color{black}\,+\,c_1}\mathbf{v}_{t}^g{\color{black}\,+\,(c_2-c_0)} \mathbf{v}_{t}^g\|
+\|\alpha\nabla F_i(\mathbf{w}_{i,t})-\alpha\nabla F(\mathbf{w}^g_{t})\|\nonumber\\
&{\color{black}= |c_0-c_1-c_2|}\| \mathbf{v}_{i,t}-\mathbf{v}_{t}^g\|
+\alpha\|\nabla F_i(\mathbf{w}_{i,t})-\nabla F(\mathbf{w}^g_{t})\|.\label{eq:dersp}
\end{align}

Given the definitions of gradients at each local workers and the genie worker $
\nabla F_i(\mathbf{w}_{i,t})= \sum_{c=1}^Cp_i(y=c)\nabla\mathbb{E}_{\mathbf{x}|y=c}[f_c(\mathbf{x},\mathbf{w}_{i,t})]
$
and
$
\nabla F(\mathbf{w}^g_{t})= \sum_{c=1}^Cp(y=c)\nabla\mathbb{E}_{\mathbf{x}|y=c}[f_c(\mathbf{x},\mathbf{w}^g_{t})],
$ respectively,
we have
\begin{align}
&\|\nabla F_i(\mathbf{w}_{i,t})-\nabla F(\mathbf{w}^g_{t})\|
\nonumber\\
&=\|\sum_{c=1}^Cp_i(y=c)\nabla\mathbb{E}_{\mathbf{x}|y=c}[f_c(\mathbf{x},\mathbf{w}_{i,t})]
-\sum_{c=1}^Cp(y=c)\nabla\mathbb{E}_{\mathbf{x}|y=c}[f_c(\mathbf{x},\mathbf{w}^g_{t})]\|
\nonumber\\
&= \|\sum_{c=1}^Cp_i(y=c)\nabla\mathbb{E}_{\mathbf{x}|y=c}[f_c(\mathbf{x},\mathbf{w}_{i,t})]
-\sum_{c=1}^Cp_i(y=c)\nabla\mathbb{E}_{\mathbf{x}|y=c}[f_c(\mathbf{x},\mathbf{w}^g_{t})]
\nonumber\\
&+\sum_{c=1}^Cp_i(y=c)\nabla\mathbb{E}_{\mathbf{x}|y=c}[f_c(\mathbf{x},\mathbf{w}^g_{t})]
-\sum_{c=1}^Cp(y=c)\nabla\mathbb{E}_{\mathbf{x}|y=c}[f_c(\mathbf{x},\mathbf{w}^g_{t})]\|
\nonumber\\
&\leq \|\sum_{c=1}^Cp_i(y=c)(\nabla\mathbb{E}_{\mathbf{x}|y=c}[f_c(\mathbf{x},\mathbf{w}_{i,t})]
-\nabla\mathbb{E}_{\mathbf{x}|y=c}[f_c(\mathbf{x},\mathbf{w}^g_{t})])\|
\nonumber\\
&+\|\sum_{c=1}^C(p_i(y=c)-p(y=c))\nabla\mathbb{E}_{\mathbf{x}|y=c}[f_c(\mathbf{x},\mathbf{w}^g_{t})]\|.
\label{eq:derGra}
\end{align}

Letting $f_{max}(\mathbf{w}^g_{t})=\max\{\nabla\mathbb{E}_{\mathbf{x}|y=c}[f_c(\mathbf{x},\mathbf{w}^g_{t})]\}_{c=1}^C$,
and applying the Lipschitz continuity, the equality of \eqref{eq:derGra} can be further rewritten as
\begin{align}
&\|\nabla F_i(\mathbf{w}_{i,t})-\nabla F(\mathbf{w}^g_{t})\|
\nonumber\\
&\leq\sum_{c=1}^Cp_i(y=c)L_c\|\mathbf{w}_{i,t}
-\mathbf{w}_{t}\|
+f_{max}(\mathbf{w}^g_{t})\sum_{c=1}^C\|(p_i(y=c)-p(y=c))\|.\label{eq:derGra2}
\end{align}

Combining \eqref{eq:derwe}, \eqref{eq:dersp}, and \eqref{eq:derGra}, we have
\begin{align}
\|\mathbf{w}_{i,t+1}-\mathbf{w}^g_{t+1}\|
&\leq \|\mathbf{w}_{i,t}-\mathbf{w}^g_{t}\|+\|\mathbf{v}_{i,t+1}-\mathbf{v}^g_{t+1}\|
\nonumber\\
&\leq \|\mathbf{w}_{i,t}-\mathbf{w}^g_{t}\|+{\color{black}|c_0-c_1-c_2|}\| \mathbf{v}_{i,t}-\mathbf{v}_{t}^g\|
+\alpha\|\nabla F_i(\mathbf{w}_{i,t})-\nabla F(\mathbf{w}^g_{t})\|
\nonumber\\
&\leq \|\mathbf{w}_{i,t}-\mathbf{w}^g_{t}\|+{\color{black}|c_0-c_1-c_2|}\| \mathbf{v}_{i,t}-\mathbf{v}_{t}^g\|
+\alpha \sum_{c=1}^Cp_i(y=c)L_c\|\mathbf{w}_{i,t}
-\mathbf{w}^g_{t}\|
\nonumber\\
&
+\alpha f_{max}(\mathbf{w}^g_{t})\sum_{c=1}^C\|p_i(y=c)-p(y=c)\|
\nonumber\\
&=\left(1+\alpha \sum_{c=1}^Cp_i(y=c)L_c\right)\|\mathbf{w}_{i,t}
-\mathbf{w}^g_{t}\|+{\color{black}|c_0-c_1-c_2|}\| \mathbf{v}_{i,t}-\mathbf{v}_{t}^g\|
\nonumber\\
&
+\alpha f_{max}(\mathbf{w}^g_{t})\sum_{c=1}^C\|p_i(y=c)-p(y=c)\|.\label{eq:derco}
\end{align}

Letting $\beta=1+\alpha \sum_{c=1}^Cp_i(y=c)L_c$, we rewrite \eqref{eq:derco} as
\begin{align}
\|\mathbf{w}_{i,t+1}-\mathbf{w}^g_{t+1}\|
&\leq \beta\|\mathbf{w}_{i,t}
-\mathbf{w}^g_{t}\|+{\color{black}|c_0-c_1-c_2|}\| \mathbf{v}_{i,t}-\mathbf{v}_{t}^g\|
\nonumber\\
&
+\alpha f_{max}(\mathbf{w}^g_{t})\sum_{c=1}^C\|p_i(y=c)-p(y=c)\|
\nonumber\\
&\leq \beta^2\|\mathbf{w}_{i,t-1}
-\mathbf{w}^g_{t-1}\|+\beta {\color{black}|c_0-c_1-c_2|}\| \mathbf{v}_{i,t-1}-\mathbf{v}_{t-1}^g\|
\nonumber\\
&
+\beta\alpha f_{max}(\mathbf{w}^g_{t})\sum_{c=1}^C\|p_i(y=c)-p(y=c)\|+{\color{black}|c_0-c_1-c_2|}\| \mathbf{v}_{i,t}-\mathbf{v}_{t}^g\|
\nonumber\\
&
+\alpha f_{max}(\mathbf{w}^g_{t})\sum_{c=1}^C\|p_i(y=c)-p(y=c)\|
\nonumber\\
&
\leq \beta^{t+1}\|\mathbf{w}_{i,0}
-\mathbf{w}^g_{0}\|+{\color{black}|c_0-c_1-c_2|}\sum_{j=0}^t\beta^{t-j}\| \mathbf{v}_{i,j}-\mathbf{v}_{j}^g\|
\nonumber\\
&
+\alpha \sum_{c=1}^C\|p_i(y=c)-p(y=c)\|\sum_{j=0}^tf_{max}(\mathbf{w}^g_{j}).
\end{align}

Hence, the proof is completed.
\end{proof}
\end{appendices}
\bibliographystyle{IEEEtran}
\bibliography{ref}
\end{document}